\documentclass[a4paper,UKenglish,cleveref,autoref,thm-restate]{lipics-v2021}
\pdfoutput=1

\hideLIPIcs
\nolinenumbers
\hypersetup{hidelinks}
\newcommand{\Nbb}{\ensuremath{\mathbb{N}}}
\newcommand{\Pbb}{\ensuremath{\mathbb{P}}}
\newcommand{\Acal}{\ensuremath{\mathcal{A}}}
\newcommand{\intn}[2]{\ensuremath{\{#1,\dots,#2\}}}

\title{On Periodic and Aperiodic Optimal Strategies in Solvency Games}
\titlerunning{Optimal Strategies in Solvency Games}

\author{Quentin Guilmant}{IRIF, Universit\'e Paris Cit\'e, France}{}{}{}
\author{Florian Luca}{Mathematics Division, Stellenbosch University, South Africa \and
  Max Planck Institute for Software Systems, Saarland Informatics Campus, Germany}{}{}{}
\author{Richard Mayr}{School of Informatics, University of Edinburgh, UK}{}{}{}
\author{Jo\"el Ouaknine}{Max Planck Institute for Software Systems,
  Saarland Informatics Campus, Germany}{}{}{}
\author{James Worrell}{Department of Computer Science, Oxford
  University, UK}{}{}{}

\authorrunning{Q.~Guilmant, F.~Luca, R.~Mayr, J.~Ouaknine, J.~Worrell}
\Copyright{Q.~Guilmant, F.~Luca, R.~Mayr, J.~Ouaknine, J.~Worrell}

\begin{document}

\begin{titlepage}
\maketitle
\begin{abstract}
Solvency games are a gambling problem on infinite-state Markov
decision processes in which the state $n \in \Nbb$ represents an investor's fortune.
In every round, the investor chooses an action from a finite action set,
and every action yields a distribution over integer-valued gains
in an interval $\intn{-\ell}m$.
The risk-averse investor wants to minimise the
probability of eventual ruin (reaching a fortune $\le 0$).

It was shown in~\cite{BKSV2008}
that memoryless deterministic optimal strategies exist,
but they are not eventually constant in general.
Even in the special case of gains in $\intn{-2}1$, the optimal strategy 
may need to make use of two different actions at arbitrarily high fortunes.

We show that optimal strategies in solvency games need not
be ultimately periodic in general (thus disproving a 2012 conjecture of Ku\v{c}era~\cite[Section~3.6]{Kucera:RP2012}).
Already in the case of gains in $\intn{-3}1$, it is possible
for the optimal strategy to be unique but aperiodic.

For gains in $\intn{-2}1$, there always exists an ultimately periodic
optimal strategy whose tail is constant or alternates between two actions.

Finally, we show that the optimal strategy is computable if it is unique.
Moreover, (some) optimal strategy can always be computed
in the case of gains in $\intn{-\ell}1$ for any $\ell \in
\Nbb$. Computability in the general case however remains open.
\end{abstract}
\end{titlepage}

\section{Introduction}\label{sec:intro}
\subparagraph{Background.}
Decision making under uncertainty is a fundamental problem
studied in computer science, operations research, and game theory.
We study solvency games, a very simple model in which a single player
(aka the investor) wants to maximise the probability of perpetual
solvency, i.e., to minimise the probability of eventual ruin.
Even in this very simple setting, optimal strategies may need to
be rather complex.

A solvency game can be described as a gambling problem on an
infinite-state MDP\@. The states of this MDP are numbers $n \in \Nbb$
which correspond to the current fortune of the investor.
In every round of the game, the investor chooses an action from a finite action set
$\Acal = \{A,B,\dots\}$.
Each action yields a distribution over integer-valued gains
in a bounded interval $\intn{-\ell}m$.
The new fortune is then the old fortune plus the gain.
The risk-averse investor wants to minimise the
probability of ruin (eventually reaching a fortune $\le 0$).

There always exist
optimal deterministic strategies that depend only on the current
fortune (i.e., they are memoryless)~\cite{BKSV2008}.
Such strategies can be described by functions $\sigma: \Nbb_{>0} \rightarrow \Acal$.
While having a large fortune can reduce the risk of ruin,
maximising the expected fortune in the long run does \textbf{not}
necessarily coincide with minimising the risk of eventual ruin.
For example, suppose action $A$ yields gains $+10$ and $-1$, each with
probability $1/2$, while action $B$ yields gains $+1$ and $-1$ with
probabilities $3/4$ and $1/4$, respectively. The expected gain of $A$
is $9/2$, compared with $1/2$ for $B$. Nevertheless, from fortune $1$,
choosing $A$ causes immediate ruin with probability $1/2$, whereas
always choosing $B$ gives an eventual ruin probability of only $1/3$.
An optimal strategy must therefore choose $B$ at fortune $1$, despite
its smaller expected gain.

Still, it seemed plausible that, at least at fortunes $n$ \textbf{above} some finite
threshold $n_0$, an optimal strategy $\sigma$ could always play an action with maximal
expected gain, i.e., $\sigma(n)=A$ for all $n \ge n_0$.
Such strategies are also called \textbf{pure-tail} or \textbf{rich-man's}
strategies.
It is easy to see that there always exist $\varepsilon$-optimal pure-tail
strategies for every $\varepsilon >0$, and that such strategies are
straightforward to elicit.
If at least one action $A$ yields a strictly positive expected
gain then some strategy that plays $A$ at all fortunes $n \ge n_0$
(where $n_0$ depends on $\varepsilon$) can be $\varepsilon$-optimal
everywhere, since the risk of eventual ruin from fortune $n$ converges to $0$
as $n \rightarrow \infty$~\cite[Fact~3]{BKSV2008}.
Otherwise, if all actions have an expected gain $\le 0$, then eventual ruin
happens almost surely for every strategy
(except in the degenerate case where some action has exactly gain $0$ only).
For a quantitative analysis of these approximation properties see~\cite{BBEK:IC2013}.

However, optimal strategies do not have the same
properties as $\varepsilon$-optimal strategies.
In particular, optimal pure-tail strategies
do not always exist~\cite{BKSV2008}.
Surprisingly, even in the special case of
two actions and gains in $\intn{-2}1$, the optimal strategy 
may need to use both actions at arbitrarily high fortunes,
even though one action has a sub-optimal expected gain.
The counterexample from~\cite{BKSV2008} does not show whether the two actions appear in any
regular pattern for increasing fortunes $n$ in the optimal strategy.
Ku\v{c}era conjectured in~\cite[Section~3.6]{Kucera:RP2012} that
there always exist optimal strategies
that are ultimately periodic, i.e., eventually repeat a finite pattern
of actions for increasing fortunes $n$. From an algorithmic
standpoint, this conjecture has one key consequence: as
we note in~\Cref{subsec:exact_algorithm},
given any fixed ultimately periodic optimal strategy $\sigma$, for
any fortune $n$, the optimal move $\sigma(n)$ can be computed in time
polynomial in the bitsize of $n$.

\subparagraph{Our contributions.}
Our main contribution is to disprove Ku\v{c}era's conjecture.
Already in the case of gains in $\intn{-3}1$
and action sets of size two, optimal strategies may need to be
significantly more complex.
In \Cref{sec:aperiodic} we construct an example where the
optimal strategy is unique but aperiodic. In general, aperiodic
optimal strategies for $(3,1)$-solvency games are intimately 
connected to orbits of certain piecewise-linear discrete dynamical
systems, whose behaviour is notoriously difficult to predict. As a
result, we conjecture that optimal strategies of fixed solvency games cannot
in general be computed in polynomial time (as a function of the
bitsize of the input fortune).\footnote{We note that the aperiodic
  optimal strategy that we exhibit in~\Cref{sec:aperiodic} \emph{does}
  admit a polynomial-time algorithm for computing the optimal move at
  fortune $n$ (as a function of the bitsize of $n$), as we show
  in~Appendix~\ref{app:aperiodic-polytime}. Nevertheless, this
  agreeable state of affairs derives from some particularly favourable
  characteristics of the optimal strategy, which cannot be expected to
  hold in general.}

For gains in $\intn{-2}1$, every game does have an ultimately periodic optimal
strategy whose tail is constant or alternates between two actions,
regardless of the size of the action set. When all actions have positive
expected gain and share the same primary root, an optimal strategy can
be chosen to be periodic with period at most $2$ from the outset;
see~\Cref{sec:21solvency}. The primary root measures the exponential decay
of ruin probabilities under repeated use of an action; it is defined in
\Cref{sec:defs}.

Finally, we present some results about the computability of optimal
strategies, i.e., finding a computable function $\sigma: \Nbb_{>0} \rightarrow \Acal$
that corresponds to a memoryless deterministic optimal strategy for a
given solvency game.
In the special case where this optimal strategy $\sigma$ is unique,
$\sigma(n)$ is computable. This is due to the fact that, for every $n$,
the value of the solvency objective at fortune $n$ can be effectively
approximated arbitrarily closely~\cite{BBEK:IC2013}.
However, this method does not yield any complexity bounds.

Even when the memoryless deterministic
optimal strategy is not unique,
it is possible to find some optimal strategy $\sigma: \Nbb_{>0} \rightarrow \Acal$
such that $\sigma(n)$ is computable for every $n \in \Nbb_{>0}$,
provided that the gains are in $\intn{-\ell}1$ for some $\ell \in \Nbb$.
Here $\sigma(n)$ can be computed in time polynomial in $n$ (rather
than the bitsize of $n$).

Computability of optimal strategies remains an open question
in the general case with gains in $\intn{-\ell}m$ for $m>1$.

\subparagraph{Related work.}
Many works in gambling theory also consider the problem of minimising
the risk of ruin, e.g.,~\cite{Browne:1995,DubbinsSavage:2014,Ferguson:1965,Hill:99,Truelove:1970}.
However, their models differ from ours. Instead of choosing between different
actions, their investor chooses which part of his or her fortune to gamble on a fixed
action. Additionally, some other models allow borrowing
or pay interest on unused capital~\cite{brazdil_et_al:LIPIcs.FSTTCS.2013.487}.

More closely related to our work are finite-state Markov decision processes
with integer rewards~\cite{theoretics:11035,Puterman:book}.
One can interpret the total reward (the sum of all
rewards so far in the run) as the investor's fortune.
Alternatively, one can consider the fortune as part of the (now infinite)
state space of an MDP\@. In other words, a state of the MDP is described by a pair $(s,n)$
where $s$ is one of finitely many control states and $n \in \Nbb$ is the
fortune. Using terminology from automata theory, this is called a \textbf{one-counter
MDP}~\cite{ajdarow_et_al:LIPIcs.ICALP.2025.138,BBEK:IC2013}.
This model is strictly more general than ours, due to the additional
control states, i.e., our model corresponds to the subclass of
one-counter MDPs with just one control state.

\subparagraph{Statement on AI use.}
The main results of this paper, specifically~\Cref{thm:main-aperiodic},
\Cref{cor:21-ultimately-periodic}, \Cref{thm:compute-upward-one}, and
\Cref{thm:unique-computable}, were obtained in 2025 without any AI
assistance. In revising this work for the present paper, we used
OpenAI's GPT-6 Astra to assist with reorganising, streamlining, and simplifying the
exposition; checking, and occasionally correcting, mathematical arguments, including
the aperiodicity construction and proof, the general periodicity result for $(2,1)$-solvency games, and
the algorithm and proof for polynomial-time evaluation of the
aperiodic strategy in Appendix~\ref{app:aperiodic-polytime}; generating and running code to check
algebraic calculations and finite prefixes of strategy sequences; and
locating, checking, and standardising bibliographic
references. Feedback generated by Anthropic's Fable~5.1 also contributed
to identifying a missing justification concerning the moduli of
characteristic roots. The authors take responsibility for the
correctness, originality, and presentation of the final paper.

\section{Solvency games and optimal strategies}\label{sec:defs}

\subsection{The model and the objective}

We first make the gambling model precise. We write
$\Nbb=\{0,1,2,\ldots\}$ and $\Nbb_{>0}=\{1,2,\ldots\}$.

\begin{definition}\label{def:solvency-game}
An \textbf{$(\ell,m)$-solvency game}, where $\ell,m\in\Nbb$, has a
finite nonempty set $\Acal$ of actions. Each action $A$ specifies a
probability distribution $\Pbb_A$ on the integer gains
$\{-\ell,\ldots,m\}$.

At a positive fortune $n$, the player chooses an action $A$ and receives
a gain $k$ drawn from $\Pbb_A$. The new fortune is $n+k$, unless this
is nonpositive, in which case the player is ruined and play ends.
We record all ruined fortunes as state $0$; thus the state update is
$n\mapsto\max\{0,n+k\}$ and the state space is $\Nbb$.
\end{definition}

The bounds $\ell$ and $m$ need not be attained by every action.
Probabilities may be real numbers; for our algorithmic results, we
assume that they are rational and given explicitly.

A \textbf{strategy} is a rule for choosing actions. It may use the whole
history of play and may randomise between actions. For a strategy
$\sigma$, let $p^\sigma(n)$ be the probability of eventual ruin when
starting with fortune $n$ and following $\sigma$. The best (smallest) possible
ruin probability is
\[
p_{\mathrm{opt}}(n):=\inf_\sigma p^\sigma(n).
\]
We seek an \textbf{optimal} strategy: a single rule that attains this
value from every positive fortune. For convenience, we put
$p^\sigma(n)=p_{\mathrm{opt}}(n)=1$ whenever $n\leq0$.
For example, a gain of $-2$ at fortune $1$ contributes a ruin
probability of $p_{\mathrm{opt}}(-1)=1$. These boundary values let us
use $n+k$ in formulas even when a loss takes us below zero.

\subparagraph{Looking one round ahead.}
Suppose we choose action $A$ at fortune $n$. With probability
$\Pbb_A(k)$, the gain is $k$ and the new fortune is $n+k$.
The best future ruin probability from there is $p_{\mathrm{opt}}(n+k)$.
Averaging over the possible gains gives the risk associated with
choosing $A$ first. Taking the smallest of these averages gives the
\textbf{Bellman equation}:
\begin{equation}\label{eq:bellman}
p_{\mathrm{opt}}(n)=
\min_{A\in\Acal}\sum_{k=-\ell}^{m}\Pbb_A(k)p_{\mathrm{opt}}(n+k)
\qquad(n>0).
\end{equation}
For example, choosing action $B$ from the introduction at fortune $1$
gives the expression
\[
\frac14+\frac34p_{\mathrm{opt}}(2).
\]
The first term accounts for immediate ruin; the second accounts for
reaching fortune $2$ and continuing from there.

This reasoning does not yet require an optimal strategy to exist.
After each possible gain, we can choose a continuation whose ruin
probability is as close as we wish to the corresponding infimum.
Also, randomising the first action only averages the risks of those
actions, so cannot improve on the best one.

A \textbf{memoryless deterministic} strategy simply fixes one action
for each positive fortune. It uses neither the earlier history nor
random choices, and is therefore a function
$\sigma:\Nbb_{>0}\to\Acal$. The Bellman equation suggests choosing a
minimising action at each fortune. The following argument shows that
these choices really do give an optimal strategy for the whole game.

\begin{proposition}[{cf.~\cite[Proposition~7]{BKSV2008}}]\label{prop:MD-optimal}
Any memoryless deterministic strategy that chooses a minimising
action in~\eqref{eq:bellman} at every positive fortune is optimal.
In particular, an optimal memoryless deterministic strategy exists.
\end{proposition}
\begin{proof}
There are finitely many actions, so we can choose a minimiser
$\sigma(n)$ at every $n>0$. Fix these choices. Let $q_t(n)$ be the
probability of ruin within the first $t$ rounds under $\sigma$,
with $q_t(n)=1$ for $n\leq0$.
We show that $q_t(n)\leq p_{\mathrm{opt}}(n)$ for every $t$ and $n$.

For $t=0$, this holds because a positive initial fortune is not yet
ruined. For the next round, suppose the bound already holds for $t$.
After the first gain $k$, there are $t$ rounds left, so for $n>0$,
\begin{align*}
q_{t+1}(n)
&=\sum_{k=-\ell}^{m}\Pbb_{\sigma(n)}(k)q_t(n+k)\\
&\leq\sum_{k=-\ell}^{m}\Pbb_{\sigma(n)}(k)p_{\mathrm{opt}}(n+k)
=p_{\mathrm{opt}}(n).
\end{align*}
The inequality uses the bound for $t$ rounds; the final equality uses
our choice of $\sigma(n)$ in the Bellman equation.
This proves the bound for every finite number of rounds.

Every play that ends in ruin does so after finitely many rounds.
Thus $q_t(n)$ increases to $p^\sigma(n)$ as $t$ grows, giving
$p^\sigma(n)\leq p_{\mathrm{opt}}(n)$.
The reverse inequality follows from the definition of
$p_{\mathrm{opt}}$, so $\sigma$ is optimal.
\end{proof}

Henceforth, \emph{strategy} means a memoryless deterministic strategy.
We view it as the sequence $\sigma(1)\sigma(2)\ldots$, indexed by
\emph{fortune}. For example, $\mathit{ABABAB}\ldots$ chooses $A$ at odd fortunes
and $B$ at even fortunes. A strategy has a \textbf{pure tail} if this
sequence is eventually constant. More generally, it is
\textbf{ultimately periodic} if some finite block repeats from some
point onwards: there are integers $d,n_0\geq1$ such that
$\sigma(n+d)=\sigma(n)$ for all $n\geq n_0$.
We call it \textbf{aperiodic} if it is not ultimately periodic.

\subsection{Simplifying the game}

An action with no negative gain keeps the player solvent forever.
We set aside this trivial case and assume that every action has a
negative gain with positive probability. We can then simplify the
game by removing waiting rounds and duplicate actions.

Next, consider a round with gain zero. The fortune has not changed,
so a memoryless strategy chooses the same action again. We can skip
these waiting rounds and record only the first nonzero gain. Its
distribution is
\[
\Pbb_{A'}(0)=0,\qquad
\Pbb_{A'}(k)=\frac{\Pbb_A(k)}{1-\Pbb_A(0)}\quad(k\neq0).
\]
The division preserves the relative probabilities of the nonzero
gains and makes their total probability $1$. For instance, an action
with probabilities $1/8,1/2,3/8$ for gains $-1,0,+1$ becomes an action
with probabilities $1/4,3/4$ for gains $-1,+1$.
A negative gain is possible, so $\Pbb_A(0)<1$: with probability $1$,
the waiting eventually ends.

Finally, if two actions now have the same distribution, we keep just
one of them. Removing waiting rounds and duplicate action names
preserves the ruin probabilities of corresponding memoryless
strategies, and hence the optimal values. To use a strategy in the
original game, we simply replace each retained action by one of the
original actions it represents.
From now on, $\Acal$ denotes the resulting action set: its distributions
are distinct and have no zero gains. Any claim of uniqueness refers
to this simplified game.

\subparagraph{Expected gains.}
The \textbf{expected gain}, also called the \textbf{drift}, of action $A$ is
\[
\mu_A:=\sum_{k=-\ell}^{m}k\Pbb_A(k).
\]
As is classically known in gambling theory, if every action has nonpositive expected gain, ruin occurs with
probability $1$ under every strategy. To see the idea, start at
fortune $n$ and stop at ruin or on reaching a fortune of at least $L>n$.
Retain any final negative fortune. Over any finite number of rounds,
the expected stopped fortune is at most $n$. If $q$ is the probability
of reaching the upper level by then before ruin, this expectation is
at least $qL+(1-q)(-\ell) = qL-(1-q)\ell$, since the fortune is at least $L$ on that
event and at least $-\ell$ otherwise. Thus
\[
q\leq\frac{n+\ell}{L+\ell},
\]
also when arbitrarily many rounds are allowed. Meanwhile, play cannot
remain forever between $1$ and $L-1$. Each of the finitely many
actions permits a loss, so every block of $L$ rounds has a fixed
positive chance of losses throughout and hence ruin. Survival
therefore requires reaching or exceeding every level $L$, but the
bound above tends to zero as $L$ grows.

We henceforth assume that at least one action has positive expected
gain; other actions may still have nonpositive expected gain.
Imagine drawing independent gains from this action indefinitely,
even after a hypothetical ruin. The average gain per round converges
to its positive expected gain with probability $1$ (the law of large
numbers). The cumulative gain therefore tends to $+\infty$ and has
a finite minimum along almost every such sequence. Increasing the
initial fortune makes it more likely to cover this worst loss, so
the probability of ruin tends to zero. Optimal play does at least
as well, giving
\begin{equation}\label{eq:ruin-limit}
\lim_{n\to\infty}p_{\mathrm{opt}}(n)=0.
\end{equation}

These assumptions imply $\ell,m\ge1$: every action permits a negative
gain, and at least one action permits a positive gain.

\subsection{Primary roots and pure tails}

We next recall how polynomials help describe the risk of repeatedly
using a single action. Write $p^A(n)$ for the ruin probability when
always choosing $A$. Looking one round ahead gives
\[
p^A(n)=\sum_{k=-\ell}^{m}\Pbb_A(k)p^A(n+k)\qquad(n>0).
\]
Unlike the Bellman equation, this recurrence has no minimum: the
action is fixed.

For a concrete example, take action $B$ from the introduction, which
gains $1$ with probability $3/4$ and loses $1$ with probability $1/4$.
Writing $b_n:=p^B(n)$, the recurrence becomes
\[
b_n=\frac34b_{n+1}+\frac14b_{n-1},
\qquad\text{or equivalently}\qquad
b_{n+1}-b_n=\frac13(b_n-b_{n-1}).
\]
Thus each successive drop in ruin probability is one third of the
previous drop. These drops form a geometric series. Since $b_0=1$
and $b_n$ tends to zero, their total is $1$, giving
$b_n=(1/3)^n$. Each extra unit of fortune divides the risk by three.

This example suggests looking for sequences of the form $x^n$:
increasing the fortune by one then multiplies the value by the fixed
factor $x$. For $n>\ell$, all fortunes in the recurrence are positive.
Substituting $x^n$, with $x\neq0$, and dividing by $x^n$ gives
\[
x^n=\sum_{k=-\ell}^{m}\Pbb_A(k)x^{n+k}
\qquad\Longleftrightarrow\qquad
1=\sum_{k=-\ell}^{m}\Pbb_A(k)x^k.
\]
Multiplication by $x^\ell$ removes the negative powers. This gives
the polynomial used to describe these geometric solutions.

\begin{definition}\label{def:char-poly}
Put
\[
F_A^{\ell,m}(x):=\sum_{k=-\ell}^{m}\Pbb_A(k)x^{k+\ell}-x^\ell.
\]
The \textbf{characteristic polynomial} of action $A$ is
$\chi_A^{\ell,m}(x):=F_A^{\ell,m}(x)/a_A$, where $a_A$ is the
coefficient of the highest power of $x$ present in $F_A^{\ell,m}$.
We omit the superscripts when the bounds are clear.
\end{definition}

A negative gain contributes a nonzero term below degree $\ell$, so
$F_A$ is not the zero polynomial. Dividing by $a_A$ makes its highest
coefficient equal to $1$, i.e., the polynomial is monic. This scaling does not change the roots.
When $\Pbb_A(m)>0$, we simply have $a_A=\Pbb_A(m)$.

The number $1$ is always a root, since the probabilities sum to $1$.
It corresponds to the constant sequence $1^n=1$, which satisfies
the recurrence but does not decay towards zero. For an action with
positive expected gain, a root between $0$ and $1$ describes the
decay of the actual ruin probabilities.

\begin{proposition}[{cf.~\cite[Lemma~1]{BKSV2008}}]\label{prop:PrimaryRoot}
If $A$ has positive drift, $\chi_A$ has a unique root $c_A\in(0,1)$,
called its \textbf{primary root}. If $m=1$, every other root
$r\notin\{1,c_A\}$ satisfies $|r|<c_A$.
\end{proposition}

An elementary proof of the modulus bound is given in
Appendix~\ref{app:secondary-root-bound}.
These other roots are called \textbf{secondary roots}.
For example, the polynomial for the action $B$ above is
\[
\chi_B(x)=x^2-\frac43x+\frac13
=(x-1)\left(x-\frac13\right).
\]
Its primary root is $1/3$, exactly the factor by which the ruin
probability decreases at each successive fortune. In general, the
primary root determines the exponential rate of decay under repeated
use of an action. Smaller primary roots mean faster decay.

This gives a useful guide to choosing actions at large fortunes.
The following result says that a uniquely smallest primary root
eventually determines the optimal choice. We state it for gains
bounded above by $1$, the case used later in the paper.
\begin{theorem}[{\cite[Theorem~8]{BKSV2008}}]\label{thm:RichMan}
Consider an $(\ell,1)$-solvency game in which every action has positive
drift. If some action $A$ satisfies $c_A<c_B$ for every $B\neq A$,
then every optimal strategy chooses $A$ at all sufficiently large fortunes.
\end{theorem}

If primary roots tie, this comparison no longer distinguishes the
actions. The secondary roots can then matter, even at arbitrarily
large fortunes.

\begin{theorem}[{\cite[Theorem~9]{BKSV2008}}]\label{thm:noRichMan}
Consider a $(2,1)$-solvency game with two positive-drift actions $A,B$,
where $\Pbb_A(-2)>0$ and $\Pbb_B(-2)>0$.
If the actions have the same primary root but distinct secondary roots,
then every optimal strategy uses both actions at arbitrarily large
fortunes.
\end{theorem}

\begin{example}\label{ex:common-root}
Consider the following two actions:
\[
\begin{array}{c|ccc|c}
 & \Pbb(-2) & \Pbb(-1) & \Pbb(+1) & \text{Expected gain}\\ \hline
 A & 1/10 & 1/10 & 4/5 & 1/2\\
 B & 1/7  & 0    & 6/7 & 4/7
\end{array}
\]
Their characteristic polynomials factor as
\[
\chi_A(x)=(x-1)\left(x-\frac12\right)\left(x+\frac14\right),
\qquad
\chi_B(x)=(x-1)\left(x-\frac12\right)\left(x+\frac13\right).
\]
Both primary roots are $1/2$, so both actions have the same
exponential rate of decay. Their secondary roots, $-1/4$ and $-1/3$,
differ. By~\Cref{thm:noRichMan}, an optimal strategy cannot eventually
settle on either action, even though $B$ has the larger expected gain.
\end{example}

The negative secondary roots offer a hint about what may happen:
their powers change sign between odd and even fortunes. The theorem
itself does not specify the pattern of optimal actions; an example
with strict alternation was given in~\cite{Brozek:2010}.
To determine such patterns, we will derive a recurrence that compares
the actions fortune by fortune.

\section{A necessarily aperiodic optimal strategy}\label{sec:aperiodic}

Our main example has just two actions and gains in $\intn{-3}1$:
\begin{equation}\label{eq:aperiodic-game}
\begin{array}{c|cccc}
 & \Pbb(-3) & \Pbb(-2) & \Pbb(-1) & \Pbb(+1)\\ \hline
 A & 1/176 & 5/176 & 42/176 & 128/176\\
 B & 1/184 & 1/184 & 54/184 & 128/184
\end{array}
\end{equation}
Both actions have positive expected gain. Their characteristic
polynomials have the particularly useful factorisations
\begin{align}
\chi_A(x)&=(x-1)\left(x-\frac12\right)
                 \left(x^2+\frac18x+\frac1{64}\right),
                 \label{eq:aperiodic-polynomials}\\
\chi_B(x)&=(x-1)\left(x-\frac12\right)
                 \left(x^2+\frac1{16}x+\frac1{64}\right).
                 \notag
\end{align}
The primary roots coincide, and all secondary roots have modulus $1/8$.
The actions differ only in the angles of their two complex secondary
roots. We shall see how this difference determines the optimal action
at each fortune.

\begin{theorem}[Main result]\label{thm:main-aperiodic}
The game in~\eqref{eq:aperiodic-game} has a unique optimal strategy,
and this strategy is not ultimately periodic.
\end{theorem}


The route to the recurrence is quite general. Since every upward step
has size $1$, surviving forever requires eventually reaching each higher fortune
in turn. We first use this observation to turn the infinite-horizon
objective into a succession of finite-target problems. A short
recurrence then tells us how to extend an optimal strategy by one
fortune. For our example, removing the two common polynomial factors
reduces that recurrence to order two.

\subsection{Reaching the next fortune}\label{sec:finite-targets}

For now, consider any $(\ell,1)$-solvency game in which every action has
$\Pbb_A(1)>0$. We retain the standing assumptions of~\Cref{sec:defs},
in particular the existence of an action with positive expected gain.
Write
\[
V^\sigma(n):=1-p^\sigma(n),\qquad V(n):=1-p_{\mathrm{opt}}(n)
\]
for the survival probabilities under $\sigma$ and under optimal play.
For a target $L>n>0$, let $h_L^\sigma(n)$ be the probability of reaching
$L$ before ruin. Set $h_L^\sigma(L)=1$ and $h_L^\sigma(i)=0$ for $i\le0$.
Only the actions at fortunes $1,\ldots,L-1$ affect these probabilities.

Two observations explain why finite targets suffice. First,
$V(n)>0$ at every positive fortune: by~\eqref{eq:ruin-limit}, survival
has positive probability at a sufficiently large fortune, and
consecutive $+1$ outcomes of a positive-drift action can reach it from
any smaller fortune. Second, as in the argument about expected gains
in~\Cref{sec:defs}, play cannot remain in a bounded interval forever
without being ruined, except on an event of probability zero.
Consequently, an infinite surviving play reaches every higher fortune
with probability $1$.

\begin{lemma}\label{lem:finite-targets}
A strategy is optimal if and only if, for every target $L$, it
maximises the probability of reaching $L$ before ruin from every
fortune $n<L$.
\end{lemma}
\begin{proof}
Let $\sigma$ be optimal. Since upward steps cannot skip $L$, survival
from $n$ first requires reaching $L$ and then surviving from there.
Thus
\[
V(n)=h_L^\sigma(n)V(L).
\]
If some other strategy reached $L$ with larger probability, using it
until reaching $L$ and then switching to $\sigma$ would give a larger
survival probability. This contradicts optimality because $V(L)>0$.
Here the temporary switch is a permissible history-dependent strategy;
optimality in~\Cref{sec:defs} is with respect to all such strategies.

Conversely, suppose $\sigma$ maximises every finite-target probability.
Compare it with an optimal strategy $\sigma^*$. For each $L>n$ we have
$h_L^\sigma(n)=h_L^{\sigma^*}\!(n)$. As $L$ increases, the events of
reaching $L$ before ruin decrease to the event of survival, apart from
the null event of bounded survival. Taking limits therefore gives
$V^\sigma(n)=V^{\sigma^*}\!(n)=V(n)$.
\end{proof}

We now fix actions $A_1,\ldots,A_N$ at fortunes $1,\ldots,N$ and
consider the target $N+1$. For brevity, write $h_i$ for the probability
of reaching this target from fortune $i$. Conditioning on the first
round gives
\begin{equation}\label{eq:hitting-first-step}
h_i=\Pbb_{A_i}(1)h_{i+1}
       +\sum_{j=1}^{\ell}\Pbb_{A_i}(-j)h_{i-j}
       \qquad (1\le i\le N),
\end{equation}
with $h_i=0$ for $i\le0$ and $h_{N+1}=1$.
Because the coefficient of $h_{i+1}$ is positive, this equation lets
us compute $h_2,h_3,\ldots$ successively as multiples of $h_1$.
Those multiples are the quantities we need.

Set $Q_0:=1$ and $Q_i:=0$ for $i<0$, and define
\begin{equation}\label{eq:Q-fixed}
Q_i:=\frac{Q_{i-1}-\sum_{j=1}^{\ell}
                    \Pbb_{A_i}(-j)Q_{i-j-1}}
                {\Pbb_{A_i}(1)}
       \qquad (1\le i\le N).
\end{equation}
Rearranging~\eqref{eq:hitting-first-step} and proceeding upwards from
$i=1$ gives $h_i=Q_{i-1}h_1$. Here $h_1>0$, since $N$ consecutive upward
outcomes have positive probability. The boundary value $h_{N+1}=1$
now yields
\begin{equation}\label{eq:Q-meaning}
Q_N=\frac{1}{h_1}>0,\qquad
  h_i=\frac{Q_{i-1}}{Q_N}.
\end{equation}
In particular, $Q_N$ is the reciprocal of the probability of reaching
$N+1$ from fortune $1$. Each $Q_i$ depends only on the actions up to
fortune $i$, so it need not be recomputed when we increase the target.

Suppose that actions have already been chosen at fortunes below $n$.
For a candidate action $A$ at fortune $n$, write
\begin{equation}\label{eq:Q-candidate}
Q_n(A):=\frac{Q_{n-1}-\sum_{j=1}^{\ell}
                    \Pbb_A(-j)Q_{n-j-1}}
                  {\Pbb_A(1)}.
\end{equation}
For this fixed prefix, maximising the chance of reaching $n+1$ is
exactly the same as minimising $Q_n(A)$. We now show that making these
choices successively gives an optimal strategy, even when several
actions attain the minimum.

\begin{proposition}[Choosing actions successively]\label{prop:greedy}
Start with $Q_0=1$ and $Q_i=0$ for $i<0$. At each fortune $n\ge1$,
choose any action attaining the minimum in
\begin{equation}\label{eq:Q-optimal}
Q_n=\min_{A\in\Acal}Q_n(A).
\end{equation}
Every strategy constructed in this way is optimal. Conversely,
every optimal strategy makes a minimising choice at every step.
\end{proposition}
\begin{proof}
Every optimal strategy must make a minimising choice at fortune $n$.
Otherwise, changing only its action there would decrease $Q_n$ and
increase $1/Q_n$, its probability of reaching $n+1$ from $1$.
This would contradict~\Cref{lem:finite-targets}.

Now fix an optimal strategy $\tau$, which exists by~\Cref{prop:MD-optimal}.
The recurrence~\eqref{eq:Q-optimal} and its initial values determine
a unique numerical sequence $Q_n$, even if different actions attain
the same minimum. Thus every strategy following the rule produces
the same $Q$-values as $\tau$. By~\eqref{eq:Q-meaning}, it has the
same finite-target probabilities as $\tau$, and is therefore optimal
by~\Cref{lem:finite-targets}.
\end{proof}

\subsection{A scalar recurrence for our example}\label{sec:scalar-recurrence}

We apply this construction to~\eqref{eq:aperiodic-game}.
For action $A$, substituting its probabilities into~\eqref{eq:Q-candidate}
gives
\[
Q_{n+1}(A)=\frac{11}{8}Q_n-\frac{21}{64}Q_{n-1}
          -\frac{5}{128}Q_{n-2}-\frac{1}{128}Q_{n-3}.
\]
Moving all terms to the left gives the coefficients of $\chi_A$,
applied to consecutive $Q$-values; the same correspondence holds
for $B$. This explains how the factorisations
in~\eqref{eq:aperiodic-polynomials} can help. Their common factor is
$x^2-\tfrac32x+\tfrac12$, whose coefficients suggest the combination
\begin{equation}\label{eq:R-definition}
R_n:=Q_n-\frac32Q_{n-1}+\frac12Q_{n-2}.
\end{equation}
The initial values are $R_{-1}=0$ and $R_0=1$.
For a candidate action $X$ at fortune $n+1$, define $R_{n+1}(X)$ by
using $Q_{n+1}(X)$ in~\eqref{eq:R-definition}. The other terms are
already fixed, so minimising $Q_{n+1}(X)$ also minimises $R_{n+1}(X)$.

To see what this achieves, substitute the expression for $Q_{n+1}(A)$:
\begin{align*}
R_{n+1}(A)
 &=Q_{n+1}(A)-\frac32Q_n+\frac12Q_{n-1}\\
 &=-\frac18Q_n+\frac{11}{64}Q_{n-1}
       -\frac5{128}Q_{n-2}-\frac1{128}Q_{n-3}\\
 &=-\frac18\left(Q_n-\frac32Q_{n-1}+\frac12Q_{n-2}\right)\\
 &\qquad-\frac1{64}\left(Q_{n-1}-\frac32Q_{n-2}
                                      +\frac12Q_{n-3}\right)\\
 &=-\frac18R_n-\frac1{64}R_{n-1}.
\end{align*}
The same calculation for $B$ gives the coefficient $-1/16$ in place
of $-1/8$. Together, the two candidate recurrences are
\begin{equation}\label{eq:R-candidates}
R_{n+1}(X)=s_XR_n-\frac1{64}R_{n-1},
\qquad s_A:=-\frac18,\quad s_B:=-\frac1{16}.
\end{equation}
The new recurrence uses only two preceding values. Its coefficients
are those of the remaining quadratic factor in each characteristic
polynomial.
In particular,
\[
Q_{n+1}(A)-Q_{n+1}(B)
 =R_{n+1}(A)-R_{n+1}(B)=-\frac1{16}R_n.
\]
Thus $A$ is forced when $R_n>0$, and $B$ is forced when $R_n<0$.
A zero value would mean that both actions are optimal at the next
fortune.

Finally, put $u_n:=8^nR_n$. This positive scaling preserves signs and
removes the coefficient $1/64$ from~\eqref{eq:R-candidates}.
The complete rule is now
\begin{equation}\label{eq:u-recurrence}
u_{-1}=0,\qquad u_0=1,\qquad
u_{n+1}=\begin{cases}
-u_n-u_{n-1},&u_n\ge0,\\[2pt]
-\dfrac12u_n-u_{n-1},&u_n<0.
\end{cases}
\end{equation}
At zero the two formulae coincide. The first few terms illustrate
how their signs select the actions:
\[
\begin{array}{c|rrrrrr}
n&0&1&2&3&4&5\\ \hline
u_n&1&-1&-1/2&5/4&-3/4&-7/8\\
\sigma(n+1)&A&B&B&A&B&B
\end{array}
\]
By~\Cref{prop:greedy}, every optimal strategy chooses $A$ at fortune
$n+1$ when $u_n>0$, and $B$ when $u_n<0$; it may choose either action
only if $u_n=0$. We have therefore reduced the two parts of the main
result to two properties of this single sequence: it never vanishes,
and its signs are not ultimately periodic. We establish both of these
properties over the next two sections.

\subsection{Two return maps}\label{sec:return-maps}

We study consecutive pairs
\[
v_n:=\begin{pmatrix}u_{n-1}\\u_n\end{pmatrix},\qquad
F(x,y):=\begin{cases}
(y,-y-x),&y\ge0,\\
(y,-y/2-x),&y<0.
\end{cases}
\]
Thus $v_{n+1}=F(v_n)$.\footnote{Up to swapping coordinates, this map
belongs to the family studied by Lagarias and Rains~\cite{LagariasRains:2005}.
The proof here uses only the explicit calculations below.}
The formula changes when the second coordinate
changes sign. To simplify the analysis, we look only at visits to
the negative quadrant, where both coordinates are negative.
The first such visit is
\[
v_2=\begin{pmatrix}-1\\-1/2\end{pmatrix}.
\]
Starting from any $(x,y)$ in this quadrant, direct calculation gives
just two possibilities for the next visit:
\begin{equation}\label{eq:return-table}
\begin{array}{c|c|c|c}
\text{Condition}&\text{Actions}&\text{Steps}&\text{Next pair}\\ \hline
0<y/x<2&\mathit{BAB}&3&(x-y/2,\;x/2+3y/4)\\
y/x>2&\mathit{BAABAAB}&7&(y,\;-x+y/2)
\end{array}
\end{equation}
Here a letter records the action prescribed by the sign of the
current second coordinate. For example, the first two steps are
always
\[
(x,y)\longmapsto(y,-x-y/2)
     \longmapsto(-x-y/2,x-y/2).
\]
If $y/x<2$, the new second coordinate $x-y/2$ is negative, and one
more step gives $(x-y/2,x/2+3y/4)$, with both coordinates negative.
This proves the first row of the table. If $y/x>2$, the second
coordinate after two steps is positive; the first return instead
takes seven steps, as verified in~\Cref{app:seven-step-return}.
Neither excursion has a zero coordinate. The boundary $y/x=2$
would produce a zero after two
steps; we shall prove that our sequence never reaches it.

The useful feature of the example is now visible. Put
\begin{equation}\label{eq:return-matrix}
M:=\begin{pmatrix}0&1\\-1&1/2\end{pmatrix}.
\end{equation}
The two return maps in~\eqref{eq:return-table} are
\[
-M^2=\begin{pmatrix}1&-1/2\\1/2&3/4\end{pmatrix}
\quad\text{and}\quad M.
\]
Thus every composition of return maps is a signed positive power
of one matrix. We will use two simple properties of this matrix.

\begin{lemma}\label{lem:return-matrix}
Let $H(x,y):=x^2-xy/2+y^2$.
\begin{enumerate}
\item Both return maps preserve $H$. In particular, their orbits on
the ellipse $H=1$ are bounded.
\item No positive power of $M$ fixes or negates a nonzero vector.
\end{enumerate}
\end{lemma}
\begin{proof}
For Part~(1), substitution gives $H(M(x,y))=H(x,y)$, and $H(-x,-y)=H(x,y)$.

Part~(2) is obtained from an elementary calculation with integer matrices
and parity, given in~\Cref{app:no-fixed-vector}. Alternatively, one can reason
as follows. Observe that $M$ is conjugate to a
rotation matrix with angle $\theta := \arccos(1/4)$. (Conjugacy means
that there is an invertible coordinate-change matrix $P$ which
turns the ellipse $H = 1$ into the unit circle, and such that
$PMP^{-1}$ is a pure counterclockwise rotation by $\theta$.) By Niven's
theorem~\cite[Corollary~3.12, p.~41]{Niven:1956}, $\theta$ is an
irrational multiple of $\pi$, whence the
desired conclusion follows.
\end{proof}

\subsection{Nonvanishing and aperiodicity}\label{sec:nonvanishing}

We can now complete both parts of~\Cref{thm:main-aperiodic}.
The ellipse keeps return pairs away from the origin, while the
arithmetic property of $M$ prevents them from closing up exactly.

\begin{proof}[Proof of~\Cref{thm:main-aperiodic}]
We first show that the sequence never vanishes.
Let $e:=(1,0)^{\mathsf T}$. The first negative pair is $v_2=M^2e$.
As long as the boundary $y=2x$ is avoided, every subsequent return
pair has the form
\begin{equation}\label{eq:return-powers}
w=\varepsilon M^ke,
\qquad \varepsilon\in\{-1,1\},\quad k\ge2,
\end{equation}
because each return multiplies by either $-M^2$ or $M$.
All these pairs satisfy $H(w)=H(e)=1$.

The boundary line $y=2x$ is spanned by $M^{-1}e=(1/2,1)^{\mathsf T}$.
If a return pair in~\eqref{eq:return-powers} lay on this line, then
$M^{k+1}e=\lambda e$ for some nonzero real $\lambda$.
Preservation of $H$ would give $\lambda^2=1$, contrary to~\Cref{lem:return-matrix}.
The boundary is therefore never reached. The return calculation
shows inductively that the next return always exists and that no
coordinate vanishes between returns. Together with $u_0=1$ and
$u_1=-1$, this proves $u_n\ne0$ for all $n\ge0$.
By~\Cref{prop:greedy}, the optimal strategy is unique.

Suppose now that the signs of $u_n$ were ultimately periodic, with
period $K$. Choose a negative pair $v_N$ after periodicity has begun.
Then $v_{N+K},v_{N+2K},\ldots$ are also negative pairs, and the same
finite sequence of return maps takes each one to the next.
Its product $R$ has the form
\[
R=\varepsilon M^d\qquad(d>0),
\]
so $v_{N+jK}=R^jv_N$ for every $j\ge0$.
Every one of these pairs lies in the negative quadrant and on $H=1$.

We can now proceed in two different ways. An elementary argument
involving summing the pairs is presented in Appendix~\ref{app:aperiodicity-sums}.
One can also note,
as per the alternative proof of
Part~(2) of~\Cref{lem:return-matrix}, that $R$, like $M$, is conjugate
to a rotation by some angle that is an irrational multiple of $\pi$,
which entails that the sequence of points $v_{N+K},v_{N+2K},\ldots$
must be dense in the ellipse, and thus cannot forever remain in the negative
quadrant, yielding the required contradiction. The signs, and hence the unique
optimal strategy, are not ultimately periodic.
\end{proof}

We conclude this section by noting that aperiodic optimal strategies can
exhibit arbitrarily long periodic stretches; in the case of our main
example, if we denote by $\sigma$ the unique optimal strategy, one can
verify that $\sigma(n+341) = \sigma(n)$ for all $1 \le n \le 1432$,
whereas this equality fails at at $n = 1433$. Periodicity, or lack thereof, can therefore
not in general be settled purely by inspection.

\section{Periodic optimal strategies}\label{sec:21solvency}

The same recurrence explains why the common-primary-root case with
gains in $\intn{-2}1$ is much simpler. There is now only one secondary
root per action. Removing the two common factors leaves a first-order
recurrence, whose sign either alternates at every step or becomes
zero.

\begin{theorem}\label{thm:common-root-periodic}
 Consider a $(2,1)$-solvency game whose actions have positive expected
gain and a common primary root $c$. There exists an optimal strategy
of period at most $2$.
If every action has positive probability of gain $-2$, then the
optimal strategy is unique: it chooses the action with the smallest
secondary root at odd fortunes and the action with the largest
secondary root at even fortunes.
\end{theorem}
\begin{proof}
Without loss of generality we assume that at least two actions
are available, otherwise there trivially is only one possible strategy (periodic
with period $1$).

Write $\chi_A(x)=(x-1)(x-c)(x-r_A)$.
By~\Cref{prop:PrimaryRoot}, the secondary root satisfies $|r_A|<c$.
The constant coefficient of $\chi_A$ is both $-cr_A$ and
$\Pbb_A(-2)/\Pbb_A(1)$, so $r_A\le0$, with equality exactly when
$\Pbb_A(-2)=0$.
Distinct action distributions have distinct secondary roots: the
three roots determine the monic polynomial, and its coefficients
recover the probabilities through~\Cref{def:char-poly}.

Set $R_n:=Q_n-(1+c)Q_{n-1}+cQ_{n-2}$, with $R_0=1$.
As before, the recurrence for $Q_n$ gives
\begin{equation}\label{eq:21-R}
R_{n+1}(A)=r_AR_n,\qquad
Q_{n+1}(A)-Q_{n+1}(B)=(r_A-r_B)R_n.
\end{equation}
Hence a positive $R_n$ calls for the smallest secondary root, and a
negative $R_n$ calls for the largest.

If all secondary roots are negative, $R_n$ changes sign at every
step and never vanishes. Starting from $R_0=1$, the choices are
therefore forced and alternate as stated.

Suppose instead that some secondary root is zero.
Then the first choice uses
the smallest root, which is negative, and the second uses the largest,
which is zero. We then have $R_2=0$, and~\eqref{eq:21-R} gives $R_n=0$
for all $n\ge2$, whatever actions follow. Every subsequent choice is
therefore optimal. In particular, we may continue alternating the
first two actions to obtain a strategy of period $2$, or repeat one
action forever to obtain a pure tail.
\end{proof}

For instance, the actions in~\Cref{ex:common-root} have secondary roots
$-1/4$ and $-1/3$. The unique optimal strategy is consequently
$\mathit{BABABA}\ldots$: action $B$ at odd fortunes and action $A$ at
even fortunes. This gives the precise pattern that the primary-root
comparison alone could not determine.

The common-root hypothesis gives an explicit optimal strategy from
the outset. Without it, the initial choices can be more complicated,
but the successive-choice construction of~\Cref{prop:greedy} still
ensures a simple tail.

\begin{corollary}\label{cor:21-ultimately-periodic}
Every $(2,1)$-solvency game has an ultimately periodic optimal strategy
whose tail has period at most $2$.
\end{corollary}

The proof examines the ratios of consecutive increases in $Q_n$.
These ratios evolve by a nonincreasing function, so their odd and
even subsequences are each monotone. The preferred action can change
only at finitely many threshold values of the ratio. With ties broken
consistently, the choices at odd and even fortunes therefore eventually
stabilise. The calculation is given in
Appendix~\ref{app:21-ultimately-periodic}.

\section{Computing optimal strategies}\label{sec:computability}

We now turn from the pattern of optimal actions to the problem of
computing them. Throughout this section, all gain probabilities are
rational and are given explicitly. There are two useful routes:
the successive-choice rule gives an exact algorithm when upward gains
are bounded by $1$, while uniqueness permits an approximation-based
argument for arbitrary bounded gains.

\subsection{An exact algorithm when upward gains are bounded by one}
\label{subsec:exact_algorithm}

\begin{theorem}\label{thm:compute-upward-one}
Every rational $(\ell,1)$-solvency game has a computable optimal
strategy. Its first $n$ actions can be found using
$O(n(\ell+1)|\Acal|)$ rational arithmetic operations, and in time
polynomial in $n$ and the size of the game description.
\end{theorem}
\begin{proof}
The trivial cases from~\Cref{sec:defs} can be recognised directly:
an action that never loses gives a constant optimal strategy, and if
all expected gains are nonpositive then every strategy is optimal.
Otherwise remove zero-gain outcomes as described there.

An action $A$ with $\Pbb_A(1)=0$ can also be discarded. Indeed, if a
memoryless strategy chose it at fortune $k$, then starting from $k$
the fortune could never exceed $k$: upward steps have size at most
$1$, and the strategy would again choose $A$ each time it returned to
$k$. Bounded play is ruined almost surely, whereas optimal survival
from $k$ has positive probability. Such an action cannot occur in an
optimal strategy. Actions with nonpositive expected gain but
$\Pbb_A(1)>0$ remain eligible.

We may now apply~\Cref{prop:greedy}, breaking ties by a fixed ordering
of the actions. Computing each candidate in~\eqref{eq:Q-candidate}
takes $O(\ell+1)$ rational operations, so computing all choices up to
$n$ takes $O(n(\ell+1)|\Acal|)$ operations. Only listed nonzero
probabilities need to be included in each sum.

For completeness, the numbers used in this computation also have
controlled size. Let $D$ be the least common multiple of the denominators of the
transition probabilities after the simplifications above. Its bit
length is polynomial in the size of the game description. Along the chosen
prefix, put
\[
a_i:=D\Pbb_{\sigma(i)}(1),\qquad P_j:=\prod_{i=1}^{j}a_i.
\]
Induction in~\eqref{eq:Q-fixed} shows that the denominator of $Q_j$
divides $P_j$. Moreover, $j$ consecutive upward outcomes reach $j+1$
from $1$, so~\eqref{eq:Q-meaning} gives
\[
0<Q_j\le\frac{D^j}{P_j}.
\]
Thus $Q_j$ can be represented as $N_j/P_j$ with
$0<N_j\le D^j$ and $P_j\le D^j$. The integers involved have
$O(j\log D)$ bits; the same bound holds when evaluating any candidate
extension. Exact arithmetic therefore gives the stated polynomial
time bound. The linear dependence on $n$ counts arithmetic operations,
not bit operations, and the bound is polynomial in $n$, rather than
in the length of its binary representation.
\end{proof}

We remark that for any fixed solvency game which is known to admit an
eventually periodic optimal strategy $\sigma$, computing $\sigma(n)$
can always be carried out in time polynomial in the bitsize of
$n$. It is unclear whether this can be done even for $(3,1)$-solvency
games in general, and we conjecture that this is not the case, owing
to the existence of aperiodic optimal strategies.

It is nevertheless worth noting that, in the specific case of our main
$(3,1)$ example, the situation is particularly favourable and the
optimal move $\sigma(n)$ at fortune $n$ \textbf{is} computable in time
polynomial in the bitsize of $n$, notwithstanding aperiodicity, via an
application of Baker's theorem on linear forms in logarithms of
algebraic numbers. The two crucial points are that (i)~both return
maps in~\eqref{eq:return-table} preserve the \emph{same} ellipse
$H=1$, and moreover (ii)~the two return maps exchange two portions of
the negative-quadrant arc of the ellipse, enabling the rapid
calculation of where the orbit stands at distant steps via a
telescoping calculation. The interested reader will find the details
in Appendix~\ref{app:aperiodic-polytime}.

\subsection{Uniqueness suffices for arbitrary bounded gains}

For gains larger than $1$, reaching a higher fortune need not involve
visiting every intermediate fortune. The successive-choice construction
therefore no longer applies. Nevertheless, effective approximation
of ruin probabilities is enough when the optimum is unique.

\begin{theorem}\label{thm:unique-computable}
If a rational solvency game has a unique optimal memoryless
deterministic strategy, then that strategy is computable.
\end{theorem}
\begin{proof}
For a fortune $n>0$ and action $A$, let
\[
b_A(n):=\sum_{k=-\ell}^{m}\Pbb_A(k)p_{\mathrm{opt}}(n+k),
\]
where $p_{\mathrm{opt}}(i)=1$ for $i\le0$.
By~\eqref{eq:bellman}, an optimal action minimises $b_A(n)$.
There is exactly one minimiser: if two actions tied, choosing either
one at $n$ and Bellman-minimising actions elsewhere would give two
distinct optimal strategies by~\Cref{prop:MD-optimal}.

The values $p_{\mathrm{opt}}(i)$ can be effectively approximated to
any prescribed rational error~\cite[Theorem~3.1]{BBEK:IC2013}.
Since $b_A(n)$ is a finite weighted average of such values, it too
admits effective approximations with certified error bounds.
Approximate all the $b_A(n)$ increasingly closely until the upper
bound for one action lies below the lower bounds for all others.
The strictly positive gap between the unique minimum and the finitely
many other values guarantees that this eventually happens. Output
that action.
\end{proof}

This argument supplies no \emph{a priori} quantitative bound on how much precision is needed to
separate the actions. If the optimal strategy is not unique, ties
also prevent it from being a general selection procedure: repeated
approximation cannot by itself certify that two values are exactly
equal. Whether every rational solvency game with upward gains larger
than $1$ has a computable optimal strategy remains open, as does
deciding whether a given solvency game has a unique optimal strategy.

\clearpage
\bibliography{solvency}

@inproceedings{ajdarow_et_al:LIPIcs.ICALP.2025.138,
  author = {Michal Ajdar{\'o}w and James C. A. Main and Petr Novotn{\'y}
            and Mickael Randour},
  title = {Taming infinity one chunk at a time: Concisely represented strategies
           in one-counter {MDPs}},
  editor = {Keren Censor-Hillel and Fabrizio Grandoni and Jo{\"e}l Ouaknine
            and Gabriele Puppis},
  booktitle = {52nd International Colloquium on Automata, Languages,
               and Programming (ICALP 2025)},
  series = {Leibniz International Proceedings in Informatics (LIPIcs)},
  volume = {334},
  pages = {138:1--138:19},
  publisher = {Schloss Dagstuhl -- Leibniz-Zentrum f{\"u}r Informatik},
  address = {Dagstuhl, Germany},
  year = {2025},
  doi = {10.4230/LIPIcs.ICALP.2025.138}
}

@article{BakerWustholz:1993,
  author = {Alan Baker and Gisbert W\"ustholz},
  title = {Logarithmic forms and group varieties},
  journal = {Journal f\"ur die reine und angewandte Mathematik},
  volume = {442},
  pages = {19--62},
  year = {1993},
  doi = {10.1515/crll.1993.442.19}
}

@inproceedings{BKSV2008,
  author = {Noam Berger and Nevin Kapur and Leonard J. Schulman and Vijay V. Vazirani},
  title = {Solvency games},
  editor = {Ramesh Hariharan and Madhavan Mukund and V. Vinay},
  booktitle = {IARCS Annual Conference on Foundations of Software Technology
               and Theoretical Computer Science (FSTTCS 2008)},
  series = {Leibniz International Proceedings in Informatics (LIPIcs)},
  volume = {2},
  pages = {61--72},
  publisher = {Schloss Dagstuhl -- Leibniz-Zentrum f{\"u}r Informatik},
  address = {Dagstuhl, Germany},
  year = {2008},
  doi = {10.4230/LIPIcs.FSTTCS.2008.1741}
}

@article{BBEK:IC2013,
  author = {Tom{\'a}{\v s} Br{\'a}zdil and V{\'a}clav Bro{\v z}ek
            and Kousha Etessami and Anton{\'i}n Ku{\v c}era},
  title = {Approximating the termination value of one-counter {MDPs}
           and stochastic games},
  journal = {Information and Computation},
  volume = {222},
  pages = {121--138},
  year = {2013},
  doi = {10.1016/j.ic.2012.01.008}
}

@inproceedings{brazdil_et_al:LIPIcs.FSTTCS.2013.487,
  author = {Tom{\'a}{\v s} Br{\'a}zdil and Taolue Chen and Vojt{\v e}ch Forejt
            and Petr Novotn{\'y} and Aistis Simaitis},
  title = {Solvency {Markov} decision processes with interest},
  editor = {Anil Seth and Nisheeth K. Vishnoi},
  booktitle = {IARCS Annual Conference on Foundations of Software Technology
               and Theoretical Computer Science (FSTTCS 2013)},
  series = {Leibniz International Proceedings in Informatics (LIPIcs)},
  volume = {24},
  pages = {487--499},
  publisher = {Schloss Dagstuhl -- Leibniz-Zentrum f{\"u}r Informatik},
  address = {Dagstuhl, Germany},
  year = {2013},
  doi = {10.4230/LIPIcs.FSTTCS.2013.487}
}

@unpublished{Brozek:2010,
  author = {V{\'a}clav Bro{\v z}ek},
  title = {A solvency game with alternating actions},
  year = {2010},
  note = {Unpublished manuscript}
}

@article{Browne:1995,
  author = {Sid Browne},
  title = {Optimal investment policies for a firm with a random risk process:
           Exponential utility and minimizing the probability of ruin},
  journal = {Mathematics of Operations Research},
  volume = {20},
  number = {4},
  pages = {937--958},
  year = {1995},
  doi = {10.1287/moor.20.4.937}
}

@book{DubbinsSavage:2014,
  author = {Lester E. Dubins and Leonard J. Savage},
  title = {How to Gamble If You Must: Inequalities for Stochastic Processes},
  series = {Dover Books on Mathematics},
  publisher = {Dover Publications},
  year = {2014},
  isbn = {9780486780641},
  url = {https://store.doverpublications.com/products/9780486780641},
  note = {Edited and updated by William D. Sudderth and David Gilat}
}

@article{Ferguson:1965,
  author = {Thomas S. Ferguson},
  title = {Betting systems which minimize the probability of ruin},
  journal = {Journal of the Society for Industrial and Applied Mathematics},
  volume = {13},
  number = {3},
  pages = {795--818},
  year = {1965},
  doi = {10.1137/0113051}
}

@article{Hill:99,
  author = {Theodore P. Hill},
  title = {Goal problems in gambling theory},
  journal = {Revista de Matem{\'a}tica: Teor{\'i}a y Aplicaciones},
  volume = {6},
  number = {2},
  pages = {125--132},
  year = {1999},
  doi = {10.15517/rmta.v6i2.173}
}

@inproceedings{Kucera:RP2012,
  author = {Anton{\'i}n Ku{\v c}era},
  title = {Playing games with counter automata},
  editor = {Alain Finkel and J{\'e}r{\^o}me Leroux and Igor Potapov},
  booktitle = {Reachability Problems: 6th International Workshop,
               RP 2012, Bordeaux, France, September 17--19, 2012, Proceedings},
  series = {Lecture Notes in Computer Science},
  volume = {7550},
  pages = {29--41},
  publisher = {Springer},
  address = {Berlin, Heidelberg},
  year = {2012},
  doi = {10.1007/978-3-642-33512-9_4}
}

@article{LagariasRains:2005,
  author = {Jeffrey C. Lagarias and Eric M. Rains},
  title = {Dynamics of a family of piecewise-linear area-preserving plane maps
           {II}. {Invariant} circles},
  journal = {Journal of Difference Equations and Applications},
  volume = {11},
  number = {13},
  pages = {1137--1163},
  year = {2005},
  doi = {10.1080/10236190500273127}
}

@book{Niven:1956,
  author = {Ivan Niven},
  title = {Irrational Numbers},
  series = {The Carus Mathematical Monographs},
  number = {11},
  publisher = {Mathematical Association of America},
  year = {1956},
  doi = {10.5948/9781614440116}
}

@article{theoretics:11035,
  author = {Jakob Piribauer and Christel Baier},
  title = {Positivity-hardness results on {Markov} decision processes},
  journal = {TheoretiCS},
  volume = {3},
  pages = {1--47},
  year = {2024},
  doi = {10.46298/theoretics.24.9},
  note = {Article 9}
}

@book{Puterman:book,
  author = {Martin L. Puterman},
  title = {Markov Decision Processes: Discrete Stochastic Dynamic Programming},
  publisher = {John Wiley \& Sons},
  year = {1994},
  isbn = {9780471619772},
  doi = {10.1002/9780470316887}
}

@article{Truelove:1970,
  author = {Alan J. Truelove},
  title = {Betting systems in favorable games},
  journal = {The Annals of Mathematical Statistics},
  volume = {41},
  number = {2},
  pages = {551--566},
  year = {1970},
  doi = {10.1214/aoms/1177697095}
}

\clearpage
\appendix
\section{The modulus of secondary roots}\label{app:secondary-root-bound}

\begin{proof}[Proof of the second assertion of~\Cref{prop:PrimaryRoot}]
Assume $m=1$ and put
\[
p:=\Pbb_A(1),\qquad q_j:=\Pbb_A(-j),\qquad
b_k:=\sum_{j=k}^{\ell}q_j\quad(1\le k\le\ell).
\]
Positive drift gives $p>0$, and the possibility of a loss gives
$b_1>0$. Using $p+\sum_{j=1}^{\ell}q_j=1$, we can divide out the
known factor $z-1$:
\[
F_A(z)=(z-1)\left(pz^\ell-
                 \sum_{k=1}^{\ell}b_kz^{\ell-k}\right).
\]
Thus every nonzero root $z\ne1$ satisfies
\[
p=\sum_{k=1}^{\ell}b_kz^{-k}.
\]
In particular, this holds for the primary root $c_A$.

For $t>0$, define $h(t):=\sum_{k=1}^{\ell}b_kt^{-k}$.
This function is strictly decreasing, since all $b_k\ge0$ and
$b_1>0$. Taking absolute values in the preceding equality gives
\[
p=\left|\sum_{k=1}^{\ell}b_kz^{-k}\right|
 \le\sum_{k=1}^{\ell}b_k|z|^{-k}=h(|z|).
\]
Since $h(c_A)=p$, it follows that $|z|\le c_A$.

If $|z|=c_A$, equality holds in the triangle inequality, so all
nonzero summands $b_kz^{-k}$ have the same argument. Their sum is
the positive real number $p$, so each is positive real.
In particular, $b_1z^{-1}$ is positive real, forcing $z$ to be
positive and hence equal to $c_A$.
Every nonzero root other than $1$ and $c_A$ therefore has modulus
strictly less than $c_A$. Any zero roots satisfy this bound as well.
\end{proof}

\clearpage
\section{Supporting calculations for the aperiodic example}
\label{app:aperiodic-calculations}

\subsection{The seven-step return}\label{app:seven-step-return}

We verify the second row of~\eqref{eq:return-table}.
Suppose $x,y<0$ and $y/x>2$, so $y<2x<0$. Applying the map $F$
successively gives the following pairs and their signs:
\[
\begin{array}{c|c|c}
j&F^j(x,y)&\text{Signs}\\ \hline
0&(x,y)&(-,-)\\
1&(y,-x-y/2)&(-,+)\\
2&(-x-y/2,x-y/2)&(+,+)\\
3&(x-y/2,y)&(+,-)\\
4&(y,-x)&(-,+)\\
5&(-x,x-y)&(+,+)\\
6&(x-y,y)&(+,-)\\
7&(y,-x+y/2)&(-,-)
\end{array}
\]
The signs follow from $y<2x<0$. In particular, $-x-y/2>0$,
$x-y/2>0$, $x-y>0$, and $-x+y/2<0$.
Thus no coordinate vanishes and no earlier pair returns to the
negative quadrant. Reading the second-coordinate signs at steps
$0,\ldots,6$ gives the action word $\mathit{BAABAAB}$.
The return pair is $(y,-x+y/2)=M(x,y)$, as claimed.

\subsection{Powers of the return matrix}\label{app:no-fixed-vector}

\begin{proof}[Proof of Part (2) of~\Cref{lem:return-matrix}]
Put
\[
N:=2M=\begin{pmatrix}0&2\\-2&1\end{pmatrix}.
\]
For every $k\ge1$, all entries of $N^k$ are even except its
lower-right entry, which is odd. This holds for $N$ itself.
To check the induction step, write
\[
N^k=\begin{pmatrix}a&b\\c&d\end{pmatrix}.
\]
Then direct multiplication gives
\[
N^{k+1}=NN^k
       =\begin{pmatrix}2c&2d\\c-2a&d-2b\end{pmatrix},
\]
which preserves the stated parities. In particular, $a+d$ is odd.

Since $\det N=4$, we also have $ad-bc=4^k$.
For either $\varepsilon\in\{-1,1\}$, expanding a determinant now gives
\begin{align*}
\det(N^k-\varepsilon2^kI)
 &=(a-\varepsilon2^k)(d-\varepsilon2^k)-bc\\
 &=2^k\bigl(2^{k+1}-\varepsilon(a+d)\bigr).
\end{align*}
The expression in parentheses is odd, so the determinant is nonzero.
Consequently, $(N^k-\varepsilon2^kI)v=0$ has only the solution $v=0$.
As $N^k=2^kM^k$, this says that $M^kv=\varepsilon v$ has no nonzero
solution, for either choice of sign.
\end{proof}

\subsection{Aperiodicity by summing return pairs}\label{app:aperiodicity-sums}

We give an elementary alternative to the final part of the proof
of~\Cref{thm:main-aperiodic}. As shown there, ultimate periodicity
would yield a negative pair $v_N$ and a matrix
$R=\varepsilon M^d$, with $\varepsilon\in\{-1,1\}$ and $d>0$,
such that every pair $R^jv_N$, for $j\ge0$, lies in the negative
quadrant and on the ellipse $H=1$.

We obtain a contradiction by summing these pairs.
For $x,y<0$ with $H(x,y)=1$ we have
$(x+y)^2\ge H(x,y)=1$, and thus $x+y\le-1$.
Let $S_q:=\sum_{j=0}^{q-1}R^jv_N$. Each summand has coordinate sum
at most $-1$, so the coordinate sum of $S_q$ is less than or equal to $-q$.
These vector sums therefore cannot remain bounded.

On the other hand, multiplying $S_q$ by $I-R$ makes the intermediate
terms cancel:
\[
(I-R)S_q
 =\sum_{j=0}^{q-1}\bigl(R^jv_N-R^{j+1}v_N\bigr)
 =v_N-R^qv_N.
\]
By Part~(2) of~\Cref{lem:return-matrix}, $R$ cannot fix
a nonzero vector, so $I-R$ is invertible. Hence
\[
S_q=(I-R)^{-1}(v_N-R^qv_N).
\]
The right-hand side is bounded independently of $q$, because both
$v_N$ and $R^qv_N$ lie on the same bounded ellipse, and multiplication
by the fixed matrix $(I-R)^{-1}$ preserves boundedness.
This is a contradiction. The signs, and hence the unique
optimal strategy, are not ultimately periodic.

\clearpage
\section{Ultimate periodicity with losses of at most two}
\label{app:21-ultimately-periodic}

\begin{proof}[Proof of~\Cref{cor:21-ultimately-periodic}]
The trivial cases from~\Cref{sec:defs} admit constant optimal
strategies: either some action never loses, or every action has
nonpositive expected gain. Otherwise apply the simplifications from
that section. Every remaining action permits a loss, zero gains have
been removed, and at least one action has positive expected gain.

We can also discard actions with $\Pbb_A(1)=0$. If a memoryless
strategy chose such an action at fortune $k$, then starting from $k$
the fortune could never exceed $k$: upward steps cannot skip $k$,
and the same action would be chosen on each return there.
Bounded play is ruined almost surely, whereas optimal survival from
$k$ has positive probability. Thus no optimal strategy uses such an
action. We may now apply~\Cref{prop:greedy}, breaking ties by a fixed
ordering of the actions.

Put $D_n:=Q_n-Q_{n-1}$, so $D_0=1$ and $D_{-1}=0$, and write
\[
\alpha_A:=\frac{\Pbb_A(-1)+\Pbb_A(-2)}{\Pbb_A(1)}>0,
\qquad
\beta_A:=\frac{\Pbb_A(-2)}{\Pbb_A(1)}\ge0.
\]
Since the probabilities sum to $1$, subtracting $Q_{n-1}$
from~\eqref{eq:Q-candidate} gives
\begin{align*}
Q_n(A)-Q_{n-1}
 &=\frac{\Pbb_A(-1)(Q_{n-1}-Q_{n-2})
         +\Pbb_A(-2)(Q_{n-1}-Q_{n-3})}{\Pbb_A(1)}\\
 &=\alpha_A D_{n-1}+\beta_A D_{n-2}.
\end{align*}
The choice at fortune $n$ therefore minimises this expression, and
\[
D_n=\min_{A\in\Acal}
          (\alpha_A D_{n-1}+\beta_A D_{n-2})\qquad(n\ge1).
\]
It follows inductively that $D_n>0$ for every $n\ge0$.

We can consequently form the ratios $t_n:=D_n/D_{n-1}$ for $n\ge1$.
Their recurrence is
\[
t_1=\min_A\alpha_A,\qquad
t_{n+1}=f(t_n),\qquad
f(t):=\min_A\left(\alpha_A+\frac{\beta_A}{t}\right).
\]
The action at fortune $n+1$ is the first action attaining the
minimum in $f(t_n)$. Each expression defining $f$ is nonincreasing
on $t>0$, so $f$ is nonincreasing too: applying it reverses
inequalities. Also $f(t)\ge t_1$ for all $t>0$, giving $t_3\ge t_1$.
Applying $f$ repeatedly now shows that
\[
t_1\le t_3\le t_5\le\ldots,
\qquad
t_2\ge t_4\ge t_6\ge\ldots.
\]

Finally, the expressions for two actions $A,B$ are equal exactly when
\[
(\alpha_A-\alpha_B)t+(\beta_A-\beta_B)=0.
\]
Unless the expressions are identical, there is at most one such
positive value of $t$. There are therefore only finitely many thresholds
at which the selected action can change. Each monotone subsequence
can pass each threshold at most once. If it stays at a threshold,
the fixed rule for breaking ties selects the same action thereafter.
The chosen action thus eventually stabilises on each parity of
fortunes, proving the claimed period bound.
Replacing each retained action by a fixed original representative
preserves both optimality and this eventual periodicity.
\end{proof}

\clearpage
\section{Computing the aperiodic strategy in polynomial time}
\label{app:aperiodic-polytime}

For the fixed $(3,1)$ example of~\Cref{thm:main-aperiodic}, we give
an algorithm computing $\sigma(n)$ in time polynomial in the bitsize
of $n$. Only the fortune $n$, given in binary, is the input.

\subsection{The return map as a rotation}

Recall the matrix $M$ and the ellipse $H=1$ from
\Cref{sec:return-maps}. Put
\[
\theta:=\arccos(1/4),\qquad L:=\pi-\theta,\qquad
\alpha:=\frac{\theta}{L}.
\]
Since $\pi/3<\theta<\pi/2$, we have $0<\theta<L$ and $0<\alpha<1$.
Moreover, $\alpha$ is irrational, since $\theta/\pi$ is irrational
by~\Cref{lem:return-matrix}.

The coordinate change
\[
P:=\frac14\begin{pmatrix}\sqrt{15}&0\\1&-4\end{pmatrix}
\quad\text{satisfies}\quad
\|Pv\|^2=H(v),\qquad
PMP^{-1}=\begin{pmatrix}1/4&-\sqrt{15}/4\\
                       \sqrt{15}/4&1/4\end{pmatrix}.
\]
Thus $P$ turns the ellipse into the unit circle and $M$ into a
counterclockwise rotation by $\theta$. The image of the negative-quadrant
arc has angles between $\pi/2$ and $\pi/2+L$. Measure the angle from
$\pi/2$, writing $z\in(0,L)$ for this position on the arc.
Applying $P$ to the boundary ray $y=2x<0$ gives $z=L-\theta$.
The seven-step return $M$ adds $\theta$, whereas the three-step
return $-M^2$ adds $2\theta-\pi=\theta-L$. Hence the return map is
\begin{equation}\label{eq:polytime-rotation}
z\longmapsto
\begin{cases}
z+\theta,&0<z<L-\theta,\\
z+\theta-L,&L-\theta<z<L.
\end{cases}
\end{equation}
This is simply addition of $\theta$ modulo $L$.

The initial negative pair is $v_2=(-1,-1/2)^{\mathsf T}$, and
$Pv_2=(-\sqrt{15}/4,1/4)^{\mathsf T}$, giving $z_0=\theta$.
Consequently, after $j$ returns,
\[
z_j=L\{(j+1)\alpha\},
\]
where $\{t\}=t-\lfloor t\rfloor$ denotes the fractional part.
Irrationality excludes the endpoints and the branch boundary,
since reaching them would make $(j+1)\alpha$ or $(j+2)\alpha$ an integer.

\subsection{Locating the required action}

Define
\[
b_j:=\lfloor(j+2)\alpha\rfloor-\lfloor(j+1)\alpha\rfloor
\in\{0,1\}\qquad(j\ge0).
\]
Here $b_j=1$ precisely when the addition in
\eqref{eq:polytime-rotation} passes the endpoint $L$, selecting
the three-step excursion. By~\eqref{eq:return-table}, the strategy is
\[
\sigma=\mathit{AB}\,W_0W_1W_2\ldots,\qquad
W_j:=\begin{cases}
\mathit{BAB},&b_j=1,\\
\mathit{BAABAAB},&b_j=0.
\end{cases}
\]
The initial $\mathit{AB}$ comes from $u_0=1$ and $u_1=-1$.
Each block $W_j$ has length $7-4b_j$, so the total length $S_k$
of the first $k$ blocks telescopes to
\begin{equation}\label{eq:polytime-block-lengths}
S_k=\sum_{j=0}^{k-1}(7-4b_j)
   =7k-4\lfloor(k+1)\alpha\rfloor\qquad(k\ge0).
\end{equation}
We used $\lfloor\alpha\rfloor=0$; in particular, $S_0=0$.

For $n=1,2$ the answers are $A,B$, respectively. For $n\ge3$,
find the unique $k$ such that $S_k\le n-3<S_{k+1}$.
Then $\sigma(n)$ is the letter at position $n-3-S_k$ of $W_k$,
counting positions from zero. Since every block has length $3$ or $7$,
the sequence $S_k$ is strictly increasing and $S_n\ge3n>n-3$.
Binary search between $0$ and $n$ therefore finds $k$ using
$O(\log n)$ evaluations of~\eqref{eq:polytime-block-lengths}.
One further floor evaluation determines $b_k$, and hence $W_k$.
It remains to justify that the floors themselves can be evaluated
exactly in time polynomial in $\log n$.

\subsection{Evaluating the floors exactly}

Irrationality alone tells us that $m\alpha$ is not an integer;
we also need to know how close it can be to one. Put
\[
\lambda:=e^{i\theta}=\frac{1+i\sqrt{15}}4.
\]
This is an algebraic number and is not a root of unity.
For an integer $m\ge1$, choose an integer $h$ nearest to $m\alpha$.
Since $0<\alpha<1$, we have $0\le h\le m$. Taking principal
complex logarithms gives
\[
\Lambda:=(m+h)\log\lambda-h\log(-1)
       =i\bigl((m+h)\theta-h\pi\bigr)
       =iL(m\alpha-h).
\]
This form is nonzero, since $\alpha$ is irrational. The
Baker--W\"ustholz theorem~\cite{BakerWustholz:1993} bounds a nonzero
integer linear form in fixed logarithms of algebraic numbers below
by an inverse power of the maximum absolute value of its integer coefficients.
Here the algebraic numbers $\lambda,-1$ are fixed and the
coefficients have absolute value at most $2m$. It follows that
there are effectively computable constants $c,C>0$ such that
\begin{equation}\label{eq:polytime-floor-separation}
\min_{h\in\mathbb Z}|m\alpha-h|>c\,m^{-C}
\qquad(m\ge1).
\end{equation}

To compute $\lfloor m\alpha\rfloor$, enclose $\alpha$ in a rational
interval, doubling the requested precision until the two endpoints,
after multiplication by $m$, have the same floor.
By~\eqref{eq:polytime-floor-separation},
this happens once the interval has width less than
$c\,m^{-(C+1)}/2$. Thus $O(\log(m+1))$ bits of precision suffice.
The extra factor $m^{-1}$ accounts for multiplying the error by $m$.
The algorithm need not calculate $c$ or $C$: they bound the precision
at which its stopping test must succeed.

The required approximations of $\alpha$ can be computed in time
polynomial in their precision. For example, use
\[
\pi=6\arcsin(1/2),\qquad
\theta=\pi/2-\arcsin(1/4).
\]
The first identity is just $\sin(\pi/6)=1/2$. At either argument,
each term of the arcsine Taylor series is at most one quarter of
the preceding term, so the tail has a geometric bound.
Hence $O(b)$ rational terms give certified
$b$-bit approximations using integers of polynomial bit length.
Forming $\alpha=\theta/(\pi-\theta)$ preserves this bound, since
the denominator is bounded away from zero.

Every floor needed by the binary search has $m=O(n)$, and all
remaining integer calculations involve $O(\log n)$ bits.
There are $O(\log n)$ floor evaluations, each taking time polynomial
in $\log n$. This proves the claimed polynomial bound in the
bitsize of the input fortune.

\end{document}